\documentclass[letterpaper,twocolumn,10pt]{article}
\usepackage{usenix-2020-09}

\usepackage{hyperref}
\usepackage{amsmath}
\usepackage{cleveref}
\usepackage{adjustbox}
\usepackage{mathtools}
\usepackage{dsfont}
\usepackage{xcolor}
\usepackage{amsthm}

\usepackage{tikz}

\usepackage{macros}

\DeclareSpecialMathOperator{\bits}{bits}
\DeclareSpecialMathOperator{\size}{size}
\DeclareSpecialMathOperator{\seq}{seq}
\DeclareSpecialMathOperator{\timecert}{timecert}
\DeclareSpecialMathOperator{\commit}{digest}
\DeclareSpecialMathOperator{\certify}{certify}
\DeclareSpecialMathOperator{\sparsecommit}{iterative\_digest}
\DeclareSpecialMathOperator{\sparsecertify}{iterative\_certify}
\DeclareSpecialMathOperator{\verify}{verify}
\DeclareSpecialMathOperator{\gcommit}{digest\_vertex}
\DeclareSpecialMathOperator{\dock}{commit}
\DeclareSpecialMathOperator{\gcertify}{certificate\_vertices}
\DeclareSpecialMathOperator{\certificatepool}{certificate\_pool}
\DeclareSpecialMathOperator{\commitpool}{digest\_pool}
\DeclareSpecialMathOperator{\nextroot}{next\_root}
\DeclareSpecialMathOperator{\nextpower}{next\_power}
\DeclareSpecialMathOperator{\nextp}{next_{2}}
\DeclareSpecialMathOperator{\prevp}{prev_{2}}
\DeclareSpecialMathOperator{\nexto}{next_{3}}
\DeclareSpecialMathOperator{\prevo}{prev_{3}}
\newcommand{\determines}[0]{label-determines}

\begin{document}

\usetikzlibrary{calc}
\usetikzlibrary{fadings}

\definecolor{cCert}{RGB}{255, 201, 143}
\definecolor{cCertText}{RGB}{212, 110, 2}
\definecolor{cCertPathText}{RGB}{127, 3, 252}
\definecolor{cCertPool}{RGB}{201, 222, 255}
\definecolor{cCertPoolText}{RGB}{65, 116, 196}
\definecolor{cSkip}{RGB}{0, 120, 68}
\definecolor{myOrange}{rgb}{1.0, 0.66, 0.07}
\colorlet{myRed}{red!90!black}
\definecolor{myBlue}{rgb}{0.25, 0.0, 1.0}

\definecolor{cProofPath1}{rgb}{0.25, 0.0, 1.0}
\definecolor{cProofPath2}{rgb}{0.4, 0.0, 0.80}

\colorlet{cCertPath}{cProofPath1!25!white}
\colorlet{cCertPathText}{cProofPath1!90!black}
\colorlet{cProofPathText1}{cProofPath1!90!black}
\colorlet{cProofPathText2}{cProofPath2!90!black}

\definecolor{cPoolPath1}{rgb}{0.0, 0.3, 1.0}
\definecolor{cPoolPath2}{rgb}{0.0, 0.6, 0.9}
\definecolor{cPoolPath3}{rgb}{0.0, 0.8, 0.7}

\colorlet{cPoolPath}{cPoolPath1!25!white}
\definecolor{cPoolPathText1}{rgb}{0.0, 0.3, 0.8}
\colorlet{cPoolPathText2}{cPoolPath2!80!black}
\colorlet{cPoolPathText3}{cPoolPath3!80!black}

\begin{tikzfadingfrompicture}[name=faderight]   
    \clip (0,0) rectangle (2,2);   
    \shade[left color=black,right color=white] (0,0) rectangle (1.7,2);
    \shade[left color=white,right color=white] (1.7,0) rectangle (2.1,2);                
\end{tikzfadingfrompicture}

\tikzset{
    myedge/.style={
        arrows=->
    },
    edgePredecessor/.style={
        arrows=->,
    },
    edgeSkip/.style={
        arrows=->,
        color=cSkip,
        very thick
    },
    emphedge/.style={
        arrows=->,
        color=myRed,
        ultra thick
    },
    edgeParent/.style={
        arrows=->,
        color=black!50,
    },
    nodeItem/.style={
        color=black!50,
    },
    certPathV/.style={
        fill=cCertPath,
        rounded corners=5
    },
    certPathE/.style={
        preaction={draw,line width=6,cCertPath,line cap=round,arrows=-}
    },
    certV/.style={
        fill=cCert,
    },
    certPoolV/.style={
        fill=cPoolPath,
        rounded corners=5
    },
    certPoolE/.style={
        preaction={draw,line width=6,cPoolPath,line cap=round,arrows=-}
    },
    toProve/.style={
        font=\boldmath,
    },
    digest/.style={
        font=\boldmath,
    },
    emphnode/.style={
        font=\boldmath,
        fill=myRed,
        rounded corners=5
    },
    codevariablepi/.style={
    },
    codevariablepipp/.style={
    },
    codevariablevi/.style={
    },
    separatingline/.style={
        gray,
        dashed
    },
    stylerepairededge/.style={
        preaction={draw,line width=10,cCert,line cap=round}
    },
    styleinvalidedge/.style={
        preaction={draw,line width=10,cCert,line cap=round}
    },
    stylepinode/.style={
        fill=colorrepairededge
    },
    highlightedge/.style={
        preaction={draw,line width=10,cCertPath,line cap=round}
    },
    highlightnode/.style={
        fill=cCertPath,
        rounded corners=8
    },
}

\date{}

\title{\Large \bf SoK: Authenticated Prefix Relations ---\\
  A Unified Perspective On Relative Time-Stamping and Append-Only Logs}

\author{
{\rm Aljoscha Meyer}\\
Technical University Berlin
}

\maketitle

\begin{abstract}
Secure relative timestamping and secure append-only logs are two historically mostly independent lines of research, which we show to be sides of the same coin --- the authentication of prefix relations. From this more general viewpoint, we derive several complexity criteria not yet considered in previous literature. We define transitive prefix authentication graphs, a graph class that captures all hash-based timestamping and log designs we know of. We survey existing schemes by expressing them as transitive prefix authentication graphs, which yields more compact definitions and more complete evaluations than in the existing literature.
\end{abstract}

\section{Introduction}


Consider the problem of mapping any finite sequence to a small digest, such that for any pair of a sequence and one of its prefixes, their digests together with a small certificate unforgeably certify that one is a prefix of the other. In other words, consider the problem of finding an authenticated data structure for the prefix relation on strings.

While we are not aware of any prior work that explicitly takes this abstract viewpoint of the problem, there are several publications that tackle it through the lens of a specific use-case: secure logging~\cite{schneier1999secure}\cite{crosby2009efficient}\cite{pulls2013distributed}, accountable shared storage~\cite{li2004secure}\cite{yumerefendi2007strong}, certificate transparency~\cite{laurie2014certificate}\cite{laurie2014certificateb}\cite{rfc9162}, or data replication~\cite{ogden2017dat}\cite{tarr2019secure}.

Some fifteen years before most of these publications, investigation into secure relative timestamping -- given two events, give an unforgeable certificate that one happened before the other -- has yielded \defined{linking schemes}~\cite{buldas1998time} as a class of efficient timestamping solutions, with several specific schemes being proposed~\cite{haber1990time}\cite{buldas1998new}\cite{buldas2000optimally}. While solutions to relative timestamping need not necessarily yield solutions to prefix authentication in general, every \textit{linking scheme} does provide prefix authentication.

So there are quite a few disjoint publications working on essentially the same problem, but with sparse interreferencing and the independent invention of various metaphorically wheel-shaped devices. All these techniques are based on the observation that if some object contains a secure hash of another object, the prior must have been created after the latter. Linking schemes are a well-defined class of such solutions, other approaches only have ad-hoc descriptions and proofs of correctness. Different publications also considered different efficiency criteria, making it difficult to objectively compare several approaches.

We provide systemization on several levels. First, we give a precise definition of prefix authentication, and systematically derive a set of efficiency criteria from this definition (\cref{prefix_authentication_schemes}). Second, we define a class of digraphs which generalizes the linking schemes, can be used for prefix authentication, and includes all hash-based prefix authentication schemes that we are aware of (\cref{transitive_prefix_authentication_schemes}). And third, we survey those prior schemes, expressing them in our formalism and evaluating them by our complexity criteria, yielding the most complete comparison of such schemes so far (\cref{existing_schemes}), and revealing some flaws in the current state of the art.

We round out our presentation with an overview of related work (\cref{related_work}), some mathematical preliminaries (\cref{preliminaries}), and a conclusion (\cref{conclusion}).

\section{Related Work}\label{related_work}

Authenticated data structures~\cite{tamassia2003authenticated} are data structures which allow to supplement query results with a small \defined{certificate} of the result's validity; a \defined{verifier} with access to only a small \defined{digest} of the current state of the data structure can check whether the query was answered truthfully based on the certificate.

Secure (relative) timestamping asks for authenticated data structures for a happened-before relation; given two events, it should be possible to prove which occurred first (as opposed to both happening concurrently). Haber and Stornetta's foundational work~\cite{haber1990time} arranges events in a linked list, using secure hashes as references. The path from an event to a prior one certifies their happened-before relation.

A straightforward optimization is batching multiple events into a single \defined{round} by storing all events within a round in a Merkle tree, and linearly linking the roots of each round rather than individual events~\cite{bayer1993improving}. This is fully analogous to the blocks of a blockchain. Certificate sizes do not decrease asymptotically however.

Buldas, Laud, Lipmaa, and Villemson provide the first asymptotic improvement by adding additional hashes such that there exist paths of logarithmic length between any pair of events~\cite{buldas1998time}. After Buldas and Laud finding the solution with the shortest certificates in this class~\cite{buldas1998new}, Buldas, Lipmaa, and Schoenmakers provide \defined{threaded authentication trees}~\cite{buldas2000optimally}, an even more efficient construction in the class of acyclic graphs where the vertex for every event is reachable from all later events. This is the class we refer to as \defined{linking schemes}, and upon which our generalization builds.

Their presentation and optimality proofs rely on a notion of time-stamping rounds, the maximum number of events in a single round features in their complexity analyses. We take on a more general setting, the notion of rounds corresponds to prefix authentication for sequences of bound length. Hence, their lower bounds do not apply to our setting.

Laurie, Langley, Kasper, Messeri, and Stradling introduce \defined{certificate transparency} (CT)~\cite{laurie2014certificate}~\cite{laurie2014certificateb}~\cite{rfc9162}, a proposed Internet standard for publicly logging information about the activities of \defined{certificate authorities} (CAs). Rather than detailing the extensive background of \defined{public-key infrastructure} into which CT embeds itself, we refer to~\cite{leibowitz2021ctng} Section 2. We shall further abstract over optimization details of CT such as \defined{signed certificate timestamps} (SCTs) or \defined{maximum merge delays}. What remains is a so-called\footnote{These logs are, quite simply, \textit{not} append-\textit{only} logs. A log maintainer cannot be prevented from creating several branches, this can merely be detected after the fact. Upon detection, a log consumer must somehow deal with the misbehavior, for example, by invalidating the log. But at that point, the log has turned into an ``append-only-until-a-full-deletion'' log, which is strictly more powerful than an append-only log. No matter how forks are handled, the data structure is too powerful. Calling a data structure an append-\textit{only} log even though it is \textit{not} is careless at best and misleading at worst, especially in a security context. Hence, we exclusively talk about prefix authentication from this point on.} append-only log with the dual-purpose of certifying both a prefix relation and set membership. The CT data structure achieves the same certificate sizes as the threaded authentication trees, but is not a linking scheme, which prompts our generalization.

Various publications propose alternatives to the original CT design, reasons include adding support for certificate revocation, strengthening the trust model, or mitigating privacy concerns~\cite{laurie2012revocation}\cite{ryan2013enhanced}\cite{kim2013accountable}\cite{melara2015coniks}\cite{basin2014arpki}\cite{yu2016dtki}\cite{leibowitz2021ctng}. Other publications extend the concept of logging objects for accountability reasons to application binaries~\cite{fahl2014hey}\cite{nikitin2017chainiac}\cite{al2018contour} or arbitrary data~\cite{eijdenberg2015verifiable}\cite{pulls2015balloon}. None of them improve the underlying prefix authentication scheme.

Several publications generalize (and formalize) the requirements behind CT as the combination of authenticating an append-only property and supporting authenticated membership queries: the \defined{logging scheme}~\cite{dowling2016secure}, the \defined{dynamic list commitment}~\cite{chase2016transparency}, or the \defined{append-only authenticated dictionary}~\cite{tomescu2019transparency} fall in this category. We are effectively ``factoring out'' the prefix authentication; the other ``factor'' being authenticated set data structures, which have already received extensive treatment~\cite{naor2000certificate}\cite{goodrich2000efficient}\cite{buldas2002eliminating}\cite{papamanthou2008authenticated} on their own. Usually so, our definition requires neither a specific \textit{append} operation nor ``append-only proofs''; the general notion of prefixes suffices.

\defined{Secure scuttlebutt}~\cite{tarr2019secure} and \defined{hypercore}\cite{ogden2017dat} rely on authenticated prefix relations for efficient event replication. Secure scuttlebutt uses a linked list, hypercore has its own solution which is almost identical to the log of CT.

Some approaches to tamper-evident logging also rely on linked lists~\cite{schneier1999secure}\cite{pulls2013distributed}. In this context, Crosby and Wallach~\cite{crosby2009efficient} designed the first non-linking scheme approach to prefix authentication that we know of. Their scheme precedes the CT design by five years, but it is strictly more complex and inefficient.
 
Work on secure networked memory has also made use of linked lists for prefix authentication~\cite{li2004secure}\cite{yumerefendi2007strong}. They speak of \defined{fork-consistency}: because forks (creation of strings neither of which is a prefix of the other) cannot be authenticated, a malicious author must consistently feed updates from the same forks to the same data consumers to avoid detection.

Since a malicious data source can easily present such different views to several consumers, data consumers should exchange information amongst each other to protect against such \defined{split world} attacks. Epidemic protocols~\cite{demers1987epidemic} can be used to this end~\cite{chuat2015efficient}\cite{ietf-trans-gossip-05}. A complementary approach is to enforce (efficient) \textit{cosigning}, where a data source must present its updates to a large number of other participants for approval~\cite{syta2016keeping}.

Various authors argue that such reactive detection mechanism are insufficient, and propose a proactive approach based on enforcing global consensus by moving events onto a blockchain~\cite{tomescu2017catena}\cite{nikitin2017chainiac}\cite{al2018contour}\cite{madala2018certificate}\cite{kubilay2019certledger}\cite{wang2020blockchain}. Why an adversary with enough power to perpetually prevent communication between nodes that received incomparable views would be unable keep those nodes in distinct bubbles that each produce separate extensions of the blockchain is beyond us, but who are we to argue against the magic powers of the blockchain?

While we restrict our attention to prefix authentication schemes that exclusively rely on secure hashing, there also exist approaches based on cryptographic accumulators~\cite{singh2017certificate}.

\section{Preliminaries}\label{preliminaries}

We write $\Nz$ for the natural numbers, and $\N$ for the natural numbers without $0$. Let $N \subseteq \Nz$, and $n \in \Nz$, then $\setleq{N}{n} \defeq \set{m \in N \st m \leq n}$ and $\setgeq{N}{n} \defeq \set{m \in N \st m \geq n}$. We denote by $\bits(n)$ the unique set of natural numbers such that $\sum_{k \in \bits(n)} 2^k = n$.

We assume basic understanding of cryptographic hash functions~\cite{menezes2018handbook}, and a basic background in graph theory~\cite{west2001introduction}. In the following, $u, v$ always denote vertices, $U, X$ always denote sets of vertices, $G$ always denotes a graph on vertices $V$ with edges $E$. As we talk about directed graphs exclusively, all graph terminology (\defined{graph}, \defined{path}, etc.) refers to directed concepts. We only consider graphs without loops. Whenever we apply a concept that is defined on sets of vertices to an individual vertex, we mean the concept applied to the singleton set containing that vertex.

A \defined{directed acyclic graph} (DAG) is a graph without cycles. The (open) \defined{out-neighborhood} of $U$ in a graph $G$ is $\outs{G}{U} \defeq \set{v \in V \setminus U \st \text{there is $u \in U$ such that $(u, v) \in E$}}$. A \defined{sink} is a vertex with an empty out-neighborhood, $\sinks{G}$ denotes the sinks of $G$. $\reach_G(v)$ denotes the set of all vertices $u \in \V(G)$ such that there is a path from $v$ to $u$ in $G$, and $\reach_G(U) \defeq \bigcup_{u \in U} \reach_G(u)$.

We use \defined{sequence} and \defined{string} synonymously and assume both to always be finite. ``$\concat$'' denotes concatenation, $\preceq$ denotes the prefix relation, for $s \preceq t$, we call $s$ a \defined{prefix} of $t$ and $t$ an extension of $s$. $\epsilon$ denotes the empty sequences. For a sequence $s$, $s_i$ denotes the $i$-th sequence item, with indexing startig at $1$.

\section{Prefix Authentication Schemes}\label{prefix_authentication_schemes}

We can now state what it means to authenticate a prefix relation.

\begin{definition}[Prefix Authentication Scheme]
  A \defined{prefix authentication scheme} (\defined{PAS}) for sequences over some universe $U$ is a triplet of algorithms $\commit$, $\certify$, and $\verify$:

  \begin{itemize}
    \item $\fun{\commit}{\kleene{U}}{\kleene{\set{0, 1}}}$ maps any sequence $s$ to some bitstring $\commit(s)$, the \defined{digest} of $s$.
    \item $\partialfun{\certify}{\kleene{U} \times \kleene{U}}{\kleene{\set{0, 1}}}$ maps any pair of sequences $s \preceq t$ to some bitstring $\certify(s, t)$, the \defined{prefix certificate} of $s$ and $t$.
    \item $\verify$ takes bitstrings $d_s$, $d_s$ and $p$ and two natural numbers $len_s$ and $len_t$, and returns $\true$ if there exist sequences $s \preceq t$ of length $len_s$ and $len_t$ respectively such that $d_s = \commit(s)$, $d_s = \commit(t)$, and $p = \certify(s, t)$. Furthermore, it must be computationally infeasible to find inputs for $\verify$ that result in $\true$ otherwise.
  \end{itemize}
\end{definition}

From this definition, we can systematically derive the efficiency criteria by which to judge a PAS. Straightforward criteria are the time and space complexity of the three functions, as well as the sizes of digests and prefix certificates. Previous work often ignores some of these, especially when they are obvious in the context of that work. In comparing several different prefix authentication approaches from fully independent work, we consider it important to make all these basic criteria explicit.

A less obvious pair of criteria derives from the question of which portions of their input the algorithms actually utilize. Both $\commit$ and $\certify$ receive full sequences as input. In practice however, we are interested in schemes that compute them from only a sublinear amount of information about the sequence, say, in a peer-to-peer system where storing full sequences would impose prohibitive storage overhead.

The first such criterium asks which information is required to indefinitely keep appending items to a sequence and compute the digests of all these extensions. We require a function $\fun{\sparsecommit}{I \times U}{\kleene{\set{0, 1}} \times I}$ that maps information (of some type $I$) about a sequence $s$ and a new item $u$ to the digest of $s \concat u$ and the information (again of type $I$) about $s \concat u$. Repeatedly calling this function allows computing the digest for any sequence. Besides the computational complexity of $\sparsecommit$, we are interested in the size of the information for sequences of length $n$.

For prefix certificate computation, we wish to map any sequence to some (small) piece of information such that the prefix certificate for any two sequences can be computed from their two pieces of information. Again we are interested in the cost of these computations, and the size of the information depending on the length of the sequence. This generalizes the notion of a \textit{time certificate}~\cite{buldas2000optimally}, which is the primary focus of the timestamping literature, whereas this criterium is rarely analyzed in the discussion of logs. We call the piece of information about each sequence its \defined{positional certificate}.

\section{A Class of Solutions}\label{transitive_prefix_authentication_schemes}

We now develop a family of prefix authentication schemes that generalizes over all hash-based secure timestamping and logging schemes that we are aware of.

When some object contains a secure hash of another object, any change to the latter would invalidate the prior. A classic data structure to leverage this property is the \defined{Merkle tree}~\cite{merkle1989certified}, a rooted tree which labels leaves with a secure hash of their contained value, and which labels inner vertices with a secure hash of the concatenation of the child labels.

We generalize the idea behind Merkle trees to arbitrary DAGs. We label sinks via some function $\fun{\lfun}{\sinks(V)}{\set{0, 1}^{k}}$. For all non-sink vertices, we aggregate the labels of their out-neighbors via some hash function $\h$. In order to deterministically apply $\h$ to sets (of out-neighbors of a vertex), we assume there is some arbitrary but fixed total order $\leq$ on $V$, and define how to convert any vertex set $U$ into a unique sequence $\seq(U)$ via $\seq(\emptyset) \defeq \epsilon$, $\seq(U) \defeq \min_{\leq}(U) \concat (U \setminus \min_{\leq}(U))$. We then define

\begin{align*}
  \lbl_{\lfun, \h}(v) &\defeq \begin{cases}
    \lfun(v) &\mbox{if } v \in \sinks(V),\\
    \h(\seq(\outs{G}{v})) & \mbox{otherwise.}\\
\end{cases}
\end{align*}

For binary out-trees, this yields exactly the Merkle tree construction. We call a pair $(G = (V, E), \lbl_{\lfun, \h})$ a \defined{Merkle DAG}.
If every maximal path from $v$ intersects $U$, then $\lbl_{\lfun, \h}(v)$ can be computed from the labels of the vertices in $U$; we say that $U$ \defined{\determines{}} $v$. If $U$ \determines{} every $x \in X$, we say that $U$ \determines{} $X$. Given $U$ and $v$ such that $U$ \determines{} $v$, and some $U$-labeling $\fun{p}{U}{\set{0, 1}^{k}}$, we denote the expected label of $v$ that can be computed from $U$ and $p$ by $\lbl_{\h}^{p}(v)$. Observe that functions $p$ are practically unique for any fixed $\lbl_{\h}^{p}(v)$:

\begin{proposition}\label{uniquely_determined}
  Let $(G = (V, E), \lbl_{\lfun, \h})$ be a Merkle DAG, let $v \in V$, and $U \subseteq V$ such that $U$ \determines{} $v$. Then it is computationally infeasible to find a labeling $\fun{p}{U}{\set{0, 1}^{k}}$ such that $\lbl_{\h}^{p}(v) = \lbl_{\lfun, \h}(v)$ and $p \neq \restrict{\lbl_{\lfun, \h}}{\domain(p)}$.
\end{proposition}

\begin{proof}
  Assume it was feasible to find $p \neq \restrict{\lbl_{\lfun, \h}}{\domain(p)}$ with $\lbl_{\h}^{p}(v) = \lbl_{\lfun, \h}(v)$. Then there must have existed a vertex $w$ with $\lbl_{\h}^{p}(w) = \lbl_{\lfun, \h}(w)$ having an out-neighbor $x$ with $\lbl_{\h}^{p}(x) \neq \lbl_{\lfun, \h}(x)$. Hence two distinct inputs to $\h$ yielded the same hash, contradicting the collision resistance of $\h$.
\end{proof}

Our prefix authentication schemes will use Merkle DAGs whose sinks we each label with a secure hash of a sequence item. We say a vertex set $U$ is a \defined{commitment} to a vertex set $X$ if $X \subseteq \reach(U)$. Changing the label of any vertex in $X$ changes the label of at least one vertex in $U$. The digests of our schemes will be the labels of singleton commitments to vertices labeled by hashes of sequence items.

We say a vertex set $U$ is a \defined{tight commitment} to a vertex set $X$ if $U$ is a commitment to $X$ and $X$ \determines{} $U$.

Our prefix certificates generalize the set membership proofs of classic Merkle trees. Merkle trees offer compact set membership proofs by reconstructing the label of the root vertex from the labels of the out-neighborhood of a path to a leaf. The out-neighborhood of a union of such paths certifies membership of several leaves at once. We can generalize this to arbitrary Merkle DAGs (see \cref{example_subgraph_proof} for an example):

\begin{definition}[Subgraph Proof]
  Let $(G = (V, E), \lbl_{\lfun, \h})$ be a Merkle DAG, let $U \subseteq V$, and let $v \in V$ such that $U \subseteq \reach_G(v)$. Let $P$ be a family of paths starting in $v$ such that $U \subseteq \outsclosed{G}{P}$, and let $\fun{p}{\outs{G}{P}}{\set{0, 1}^{k}}$, where $\set{0, 1}^{k}$ is the codomain of $\h$.
  
  We then call $(\lbl_{\lfun, \h}(v), p)$ a \defined{potential subgraph proof} of $U$ for $v$.

  Observe that $\outs{G}{P}$ \determines{} $v$. We say $(\lbl_{\h}(v), p)$ is a \defined{(verified) subgraph proof} if $\lbl_{\h}^{p}(v) = \lbl_{\lfun, \h}(v)$, and a \defined{refuted subgraph proof} otherwise.
\end{definition}

\begin{proposition}\label{subgraph_proof}
  Let $(G = (V, E), \lbl_{\lfun, \h})$ be a Merkle DAG, and let $v \in V$. Then, by \cref{uniquely_determined}, it is computationally infeasible to find $U$ and $p$ such that $(\lbl_{\lfun, \h}(v), p)$ is a verified subgraph proof of $U$ for $v$ with $p \neq \restrict{\lbl_{\lfun, \h}}{\domain(p)}$.
\end{proposition}


\begin{corollary}\label{cor_subgraph_proof}
  Let $(G = (V, E), \lbl_{\lfun, \h})$ be a Merkle DAG, let $v \in V$, let $U \subseteq V$, and let $(\lbl_{\lfun, \h}(v), p)$ be a subgraph proof of $U$ for $v$. If $\h$ is secure, then $\outsclosed{G}{\domain(p)}$ is a subgraph of $G$. In particular, $\induced{G}{U}$ is a subgraph of $G$, and $U \subseteq \reach_G(v)$.
\end{corollary}

\begin{figure}
  \centering
  \includegraphics{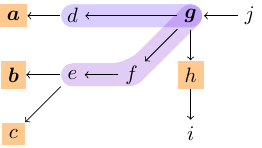}
  \caption{An example subgraph proof of $\set{a, b}$ for $g$.\\$P \defeq \set{d, e, f, g}$ consists of the vertices of \textcolor{cProofPathText1}{two} \textcolor{cProofPathText2}{paths} starting in $g$ whose out-neighborhoods together include $\set{a, b}$. The \textcolor{cCertText}{out-neighborhood of $P$ and its labels yield $p$}, and it (necessarily) \determines{} $g$: the label of $d$ can be computed from the label of $a$; the labels of $b$ and $c$ suffice to compute the label of $e$, which in turn determines the label of $f$; together with the label of $h$, we can compute the label of $g$. If this computed label matches the label given with the subgraph proof, we can be certain (up to hash collisions) that $a$ and $b$ relate to $g$ in the graph as expected and indeed have the labels claimed by $p$.}
  \label{example_subgraph_proof}
\end{figure}

\subsection{Linking Schemes}\label{linking_schemes}

We now have the terminology to define the \defined{linking schemes}, a class of prefix authentication schemes based on Merkle DAGs. We use secure hashes of sequence items to label the sinks of some Merkle DAG in which the set of the sinks that correspond to any prefix of the sequence has a common ancestor vertex; the labels of the common ancestors serve as digests, and subgraph proofs between the digest vertices serve as prefix certificates.

First, we formalize the notion of mapping sequence items to sinks:

\begin{definition}[Sequence Graph]
  Let $G$ be an acyclic graph with $\sinks(G) \supseteq \N$, and let $s$ be a sequence of length $len_s$.
  
  The \defined{sequence graph} of $s$ and $G$ is the Merkle DAG $(G, \lbl_{\lfun_s, \h})$ with

  \begin{align*}
    \lfun_s(s_v) &\defeq \begin{cases}
      \h(v) &\mbox{if } v \in \N^{< len_s},\\
      \h(\epsilon) & \mbox{otherwise.}\\
  \end{cases}
  \end{align*}
\end{definition}

Next, we can describe the class of graphs that allows for prefix authentication:

\begin{definition}[Linking Scheme Graph]
  A graph $G = (V, E)$ is a \defined{linking scheme graph} if $G$ is acyclic, $\sinks(G) \supseteq \N$, and there exist functions $\fun{\gcommit}{\N}{V}$, and $\fun{\gcertify}{\N \times \N}{\powerset{V}}$ such that for all $len_s, len_t \in \N$ with $len_s < len_t$:

  \begin{itemize}
    \item $\gcommit(len_s)$ is a tight commitment to $\setleq{\N}{len_s}$, and
    \item $\gcertify(len_s, len_t)$ is a path starting in $\gcommit(len_t)$ such that $\gcommit(len_s) \in \outsclosed{G}{\gcertify(len_s, len_t)}$.
  \end{itemize}
\end{definition}


Using a linking scheme graph as the underlying graph of a Merkle DAG yields a prefix authentication scheme (\cref{example_subgraph_scheme} gives an example):

\begin{definition}[Linking Scheme]
  Let $G = (V, E)$ be a linking scheme graph, and let $\h$ be a secure hash function.
  
  $G$ and $\h$ define a \defined{linking scheme} $(\commit, \dock, \certify)$ using the following functions:
  
  For strings $s$ of length $len_s$, let $(G, \lbl_{\lfun_s, \h})$ be the sequence graph of $s$ and $G$. We then define $\commit(s) \defeq \lbl_{\lfun_s, \h}(\gcommit(len_s))$. Observe that this can be computed from $s$ alone as $\setleq{\N}{len_s}$ \determines{} $\gcommit(len_s)$.
  
  For strings $s \preceq t$ of length $len_s$ and $len_t$ respectively, let $(G, \lbl_{\lfun_t, \h})$ be the sequence graph of $t$ and $G$. We then define $\certify(s, t)$ as the bitstring obtained by sorting $\outsclosed{G}{\gcertify(len_s, len_t)}$ according to $\leq$ and concatenating the labels (in $(G, \lbl_{\lfun_t, \h})$) of these vertices. Finally, we define $\verify(d_s, d_t, p, len_s, len_t)$ to first parse $p$ into a labeling $p'$ of $\outsclosed{G}{\gcertify(len_s, len_t)}$ and then return whether $(d_t, p')$ is a verified subgraph proof of $\gcommit(len_s)$ for $\gcommit(len_t)$.
\end{definition}

\begin{figure}
  \centering
  \includegraphics{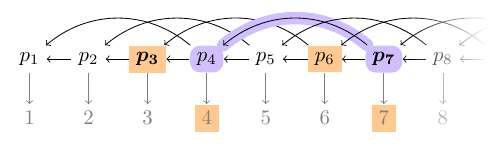}
  \caption{An example linking scheme, together with a certificate for $len_s \defeq 3$ and $len_t \defeq 7$. $p_3$ and $p_7$ are the digest vertices for $3$ and $7$ respectively. The \textcolor{cCertPathText}{path $(p_7, p_4)$} is a (shortest) path from $p_7$ whose out-neighborhood contains $p_3$, and its \textcolor{cCertText}{out-neighborhood} yields the prefix certificate.}
  \label{example_subgraph_scheme}
\end{figure}

\begin{proposition}
  \label{maintheoremLinking}
  Every linking scheme is a prefix authentication scheme.
\end{proposition}

\begin{proof}
  All functions have the required signatures. Let $s \preceq t$ be sequences of length $len_s$ and $len_t$ respectively, then $\verify(\commit(s), \commit(t), \certify(s, t), len_s, len_t)$ returns $\true$ by construction; and returning $\true$ for these inputs implies $s \preceq t$ by \cref{cor_subgraph_proof}. Any other inputs that yield $\true$ witness a hash collision by \cref{subgraph_proof}.
\end{proof}

This definition of linking schemes is adapted from Buldas et al.~\cite{buldas2000optimally} and captures the linear linking scheme~\cite{haber1990time}, anti-monotone schemes~\cite{buldas1998time}\cite{buldas1998new}, and the threaded authentication trees~\cite{buldas2000optimally}. This class of schemes does not however include Crosby and Wallach's secure logging scheme ~\cite{crosby2009efficient}, transparency logs~\cite{laurie2014certificate}, or hypercore~\cite{ogden2017dat}, prompting our search for a further generalization.

\subsection{Transitive Prefix Authentication Schemes}\label{tpass}

The generalization to include these other schemes is simple\footnote{Simple, that is, when starting from a characterization of linking schemes that was specifically crafted to allow for this generalization.}: rather than giving subgraph proofs that some digest vertex is reachable from another, we give subgraph proofs that some set that \determines{} a digest vertex is reachable from another digest vertex.

\begin{definition}[Transitive Prefix Authentication Graph]
  A graph $G = (V, E)$ is a \defined{transitive prefix authentication graph} (TPAG) if $G$ is acyclic, $\sinks(G) \supseteq \N$, and there exist functions $\fun{\gcommit}{\N}{V}$, $\fun{\dock}{\N}{\powerset{V}}$, and $\fun{\gcertify}{\N \times \N}{\powerset{V}}$ such that for all $len_s, len_t \in \N$ with $len_s < len_t$:

  \begin{itemize}
    \item $\gcommit(len_s)$ is a tight commitment to $\setleq{\N}{len_s}$,
    \item $\dock(len_s)$ \determines{} $\gcommit(len_s)$, and
    \item $\gcertify(len_s, len_t)$ is a union of paths, each starting in $\gcommit(len_t)$, such that $\dock(len_s) \subseteq \outsclosed{G}{\gcertify(len_s, len_t)}$.
  \end{itemize}
\end{definition}

Using a TPAG as the underlying graph of a Merkle DAG yields a prefix authentication scheme (\cref{example_tpas} gives an example). Because the definition is similar to that of linking schemes, we have typeset the differences in a \textbf{bold} font.

\begin{definition}[Transitive Prefix Authentication Scheme]
  Let $G = (V, E)$ be a \textbf{TPAG}, and let $\h$ be a secure hash function.
  
  $G$ and $\h$ define a \defined{transitive prefix authentication scheme} (TPAS) $(\commit, \dock, \certify)$ using the following functions:
  
  For strings $s$ of length $len_s$, let $(G, \lbl_{\lfun_s, \h})$ be the sequence graph of $s$ and $G$. We then define $\commit(s) \defeq \lbl_{\lfun_s, \h}(\gcommit(len_s))$. Observe that this can be computed from $s$ alone as $\setleq{\N}{len_s}$ \determines{} $\gcommit(len_s)$.
  
  For strings $s \preceq t$ of length $len_s$ and $len_t$ respectively, let $(G, \lbl_{\lfun_t, \h})$ be the sequence graph of $t$ and $G$. We then define $\certify(s, t)$ as the bitstring obtained by sorting $\outsclosed{G}{\gcertify(len_s, len_t)}$ according to $\leq$ and concatenating the labels (in $(G, \lbl_{\lfun_t, \h})$) of these vertices. Finally, we define $\verify(d_s, d_t, p, len_s, len_t)$ to first parse $p$ into a labeling $p'$ of $\outsclosed{G}{\gcertify(len_s, len_t)}$ and then return whether $(d_t, p')$ is a verified subgraph proof of $\gcommit(len_s)$ for $\gcommit(len_t)$ \textbf{and whether $\lbl_{\h}^{p'}(\gcommit(len_s)) = d_s$}.
\end{definition}

\begin{proposition}
  \label{maintheoremTransitive}
  Every \textbf{TPAS} is a prefix authentication scheme.
\end{proposition}

\begin{proof}
  All functions have the required signatures. Let $s \preceq t$ be sequences of length $len_s$ and $len_t$ respectively, then $\verify(\commit(s), \commit(t), \certify(s, t), len_s, len_t)$ returns $\true$ by construction; and returning $\true$ for these inputs implies $s \preceq t$ by \cref{cor_subgraph_proof}. Any other inputs that yield $\true$ witness a hash collision by \cref{subgraph_proof} (when verifying the subgraph proof) \textbf{or by \cref{uniquely_determined}} (when checking that $\lbl_{\h}^{p'}(\gcommit(len_s)) = d_s$).
\end{proof}

\begin{figure}
  \centering
  \includegraphics{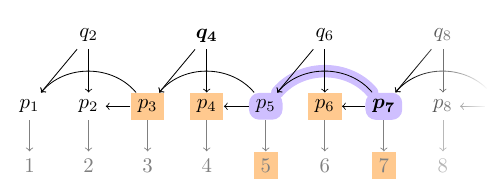}
  \caption{A TPAS, together with a certificate for $len_s \defeq 4$ and $len_t \defeq 7$. $q_4$ and $p_7$ are the digest vertices for $4$ and $7$ respectively. The \textcolor{cCertPathText}{path $(p_7, p_5)$} is a (family of exactly one) path from $p_7$ whose out-neighborhood contains $p_3$ and $p_4$, who together label-determine $q_4$. Its \textcolor{cCertText}{out-neighborhood} yields the prefix certificate.}
  \label{example_tpas}
\end{figure}

\subsection{Efficiency Criteria}\label{tpas_efficiency}

Because every TPAS stems from a TPAG, we can reason about prefix authentication schemes while remaining solely in the realm of unlabeled digraphs; the labelings are a deterministic afterthought. Given the length $k$ of each individual hash, we can also reason about the complexity criteria of \cref{prefix_authentication_schemes} by solely considering graph properties.

Since digests are labels of individual vertices, the digest size of any sequence is $k$.

Worst-case and average prefix certificate sizes correspond directly to the sizes of the out-neighborhood of $\gcertify(len_s, len_t)$: the prefix certificate is $\restrict{\lbl_{\lfun, \h}}{\outs{G}{\gcertify(len_s, len_t)}}$, which can be encoded by fixing an ordering on $\V(G)$ in advance and simply listing the labels of $\outs{G}{\gcertify(len_s, len_t)}$ according to that ordering. The space requirement is hence $\abs{\outs{G}{\gcertify(len_s, len_t)}} * k$.

Computing $\verify$ consists of reconstructing $\commit(len_s)$ from $\dock(len_s)$ and $\commit(len_t)$ from $\certify(len_s, len_t)$. Reconstructing $\commit(len_s)$ from $\dock(len_s)$ requires time proportional to the number of distinct edges on all paths from $\commit(len_s)$ to any vertices from $\dock(len_s)$. Reconstructing $\commit(len_t)$ from $\certify(len_s, len_t)$ requires time proportional to the number of edges in the graph induced by the closed neighborhood of $\gcertify(len_s, len_t)$.

The complexity of computing $\commit(len_s)$ and $\certify(len_s, len_t)$ is less straightforward. Because $\gcommit(len_s)$ is a common ancestor of all $i \leq len_s$, computing its label requires at least $\complexity{len_s}$ time. In a realistic setting, one would instead precompute and store the labels of all all digest vertices, turning the computation into a simple look-up. The additional storage cost of $\complexity{len_s}$ space does not exceed the cost for storing $len_s$ sequence items in the first place. But when also precomputing the labels for all vertices that can appear in any prefix certificate, the overall space complexity might exceed $\complexity{len_s}$.

We simplify the analysis by looking at the storage cost for precomputing the labels of \textit{all} vertices. Intuitively, we expect efficient schemes to not have redundant vertices, so this simplification should only be an overapproximation for schemes that are of little practical interest to begin with.

In order to classify the storage cost per sequence item, we define the graph $\tpaston{G}{len_s} \defeq \bigcup_{i \leq len_s} \reach(\gcommit(i))$ of all vertices that are required to work with a sequence of length $len_s$. The number of labels that need to be precomputed and stored because of the $len_s$-th sequence item is then given by $\abs{\V(\tpaston{G}{len_s}) \setminus \V(\tpaston{G}{len_s - 1})}$. We are both interested in the worst-case for any $len_s$ and in the amortized case of averaging over the first $len_s$ sequence items.

For iterative computation of digests, we require a function $\fun{\commitpool}{\N}{\V(G)}$ such that $\commitpool(len_s - 1) \cup \set{len_s}$ \determines{} both $\commit(len_s)$ and all vertices in $\commitpool(len_s)$. Intuitively, the digest pool for some $len_s$ consists of all the vertices in $\tpaston{G}{len_s}$ whose label impacts the label of any future vertex, i.e., every vertice in $\outsclosed{}{\tpaston{G}{len_t} - \tpaston{G}{len_s}} \cap \tpaston{G}{len_s}$ for any $len_t > len_s$.
The labels of $\commitpool(len_s)$ then allow for indefinitely appending new items to a sequence of length $len_s$ by adding any newly created vertices whose labels will be used in the future; we are interested in functions that minimize $\abs{\commitpool(len_s)}$. 

For computing prefix certificates from only parts of full sequences, we consider functions $\fun{\certificatepool}{\N}{\V(G)}$ such that $\outsclosed{G}{\certificatepool(len_s)} \cup \outsclosed{G}{\certificatepool(len_t)}$ \determines{} $\outsclosed{G}{\gcertify(len_s, len_t)}$. This allows for computing $\certify(len_s, len_t)$ from the (out-neighborhoods of the) \defined{certificate pools} of $len_s$ and $len_t$. We are then interested in functions that minimize $\abs{\outsclosed{G}{\certificatepool(len_s)}}$.

\subsection{Secure Timestamping}\label{secure_time_stamping}

Secure timestamping~\cite{haber1990time} asks to cryptographically certify the happened-before relation on some totally-ordered set of events. Prefix authentication does not immediately imply secure timestamping, but the problems are related: if event number $s$ happened before event number $t$, we can certify that the sequence of the first $s$ events is a prefix of the sequence of the first $t$ events. This only relates the \textit{digests} of the event sequences however, not the events \textit{themselves}.

We can extend every TPAS to provide time stamping. We define the \defined{identifier} of the $i-th$ item in some sequence as a (deterministically selected) subgraph proof of its vertex for $\gcommit(i)$. We call it an \textit{identifier} because it identifies a particular item as occuring at a particular position in a particular sequence (and its extensions).

Let $s$ be the sequence of the first $len_s$ events, and let $t$ be the sequence of the first $len_t$ events, with $s \preceq t$. To certify that event number $len_s$ happened before event number $len_t$, simply provide $\certify(s, t)$ together with the identifiers of $len_s$ and $len_t$. $\certify(s, t)$ certifies the happened-before relation of the digests, and the identifiers tie the digests to the actual items. Verification consists of verifying the certificate as well as the two identifiers.

Hence, the worst-case and average size of the identifier for any item at position $n$ becomes another complexity parameter of interest. Its exact value is $k$ (the size of an individual digest) times the number of vertices in the subgraph proof of $\set{i}$ for $\gcommit(i)$.

\section{Prior Schemes}\label{existing_schemes}

We now give definitions of timestamping and logging schemes from the literature, expressed as TPASs. This serves as a demonstration of the generality and applicability of TPASs, it provides a survey of existing approaches, and it allows us to apply our efficiency criteria to previous work.

Our definition of Merkle DAGs automatically incorporates the notion of \textit{computing} the labels along a path from a prefix certificate rather than directly \textit{using} those labels as the certificate, an optimization introduced by Buldas, Lipmaa, and Schoenmakers.~\cite{buldas2000optimally} after several prior schemes had already been published. The improvement that their section 5.2 gives over the antimonotone linking schemes~\cite{buldas1998time}\cite{buldas1998new} is inherent to our formulations of \textit{all} approaches.

Our presentation of timestamping schemes differs significantly from their original presentation in that we do not consider a setting of timestamping rounds. Working with timestamping rounds effectively amounts to solving prefix authentication for strings of bounded length. Once the maximal string length is reached, the round concludes and a new round begins for the next subsequence of bounded length.

For authentication across rounds, the rounds must themselves be maintained in a prefix-authenticating data structure. This requires an awkward nesting of prefix authentication schemes that is overall less efficient than authenticating the full string without subdividing it into rounds.

Prefix authentication for strings of bounded length is an easier problem than prefix authentication for strings of arbitrary length. Hence, we need to adapt the timestamping schemes to the more general setting, and this adaptation results in worse positional certificate sizes than the original publications report. The original publications do not account for the cost of inter-round authentication, which is why we do our own complexity analyses and arrive at worse bounds.

This adaptation also means that the proofs of optimality in the timestamping literature do not apply to our setting. While we believe our adaptations are faithful, more efficient solutions for round-less authentication can exist.

\subsection{Trivial Schemes}

The simplemost linking scheme is the \defined{linear scheme} of~\cite{haber1990time}. Its underlying graph is a ``Merkle linked-list'' $G_{lin}$:

\begin{align*}
  V_{lin} &\defeq \set{p_n \st n \in \N} \cup \N,\\
  E_{lin} &\defeq \set{(p_{n + 1}, p_n) \st n \in \N} \cup \set{(p_n, n) \st n \in \N},\\
  G_{lin} &\defeq (V_{lin}, E_{lin}).
\end{align*}

To use this graph as a linking scheme, define $\gcommit(n) \defeq p_n$, and define $\gcertify(len_s, len_t)$ as the shortest path from $\gcommit(len_t)$ to $\gcommit(len_s)$ (see \cref{fig_linear} for a depiction). We use the same definitions of $\gcommit$ and $\gcertify$ for all linking schemes in this section, unless specified otherwise.

\begin{figure}
  \centering
  \includegraphics{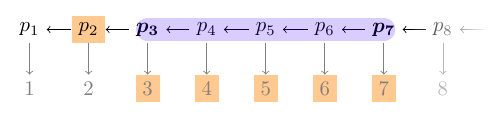}
  \caption{The linear linking scheme, highlighting \textcolor{cCertPathText}{$\gcertify(3, 7)$} and \textcolor{cCertText}{$\certify(3, 7)$}.}
  \label{fig_linear}
\end{figure}

Prefix certificates are of linear size in $len_t - len_s$. Certificate pools are of linear size in $len_s$, as they must contain the full path from $\gcommit(len_s)$ to $\gcommit(1)$: $\certificatepool(len_s) \defeq \tpaston{G_{lin}}{len_s}$. On the plus side, digest pools are of constant size, with $\commitpool(len_s) \defeq \gcommit(len_s)$.

The \defined{full linking scheme} is the other trivial scheme, with an edge from $p_j$ to $p_i$ for all $i < j$; the quadratic number of edges makes it unsuitable for any practical use.

\subsection{Skip List Schemes}

A simple but suboptimal way of interpolating between the two trivial schemes is to use a (contracted) deterministic skip list~\cite{pugh1990skip}. In addition to the edges of $G_{lin}$, also add an edge from $p_{n}$ to $p_{n - k}$ if $n$ is divisible by $k$. The certificate pool of $n$ is the out-neighborhood of the shortest path from $2^{\ceil{\log_2(n)}}$ to $n$ and from $n$ to $1$. The digest pool is the out-neighborhood of the shortest path from $n$ to $1$. \Cref{fig_skip} visualizes the construction.

\begin{figure}
  \centering
  \includegraphics{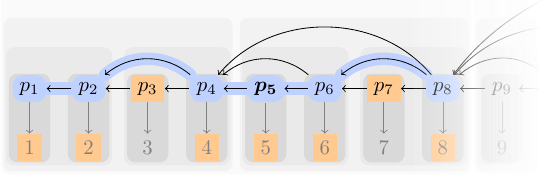}
  \caption{The CHAINIAC scheme, highlighting \textcolor{cPoolPathText1}{$\certificatepool(5)$} and its \textcolor{cCertText}{out-neighborhood}.}
  \label{fig_skip}
\end{figure}

This scheme is used by CHAINIAC~\cite{nikitin2017chainiac}, but it has prefix certificates of superlogarithmic size. Consider the vertex $n \defeq 2^k$. It has $k$ out-neighbors, all of which must occur in $\gcertify(1, 2^k)$. The second vertex of the shortest path from $2^k$ to $1$ is $2^{k - 1}$, whose out-neighborhood contains $k-1$ vertices. Iterating this argument yields a certificate size of $\Sigma_{i \leq k} i \in \complexity{k^2} = \complexity{\log(n)^2}$. While certificate pools have logarithmic size, their out-neighborhoods do not.

Blibech and Gabillon~\cite{blibech2006new} also propose a scheme based on skip lists, but their construction relies on the notion of timestamping rounds: unlike the timestamping schemes we generalize next, their scheme gives dedicated treatment to the last vertex of each timestamping round, and it can only provide prefix certificates for items whose timestamping rounds have been concluded. Hence, it is not a prefix authentication scheme according to our definition, as it is not applicable to authenticating sequences of unbounded length.

\subsection{Antimonotone Binary Schemes}

The \defined{simple antimonotone binary linking scheme}~\cite{buldas1998time} achieves logarithmic certificate pools by augmenting the linear scheme with only one additional outgoing edge per vertex. The additional edge for vertex $p_n$ goes to $p_{f_2(n)}$, with $f_2(n)$ defined as follows:

\begin{align*}
  f_2(n) &\defeq \begin{cases}
      n - (2^{k - 1} + 1) &\mbox{if } n = 2^k - 1, k \in \N \\
      n - 2^{g(n)} & \mbox{otherwise}\\
  \end{cases}\\
  g(n) &\defeq \begin{cases}
      k &\mbox{if } n = 2^{k} - 1, k \in \N \\
      g(n - (2^{k - 1} - 1)) & \mbox{if } 2^{k - 1} - 1 < n < 2^{k} - 1, k \in \N.
  \end{cases}\\
\end{align*}

Observe that for all $n < m$ we have that $f_2(n) \geq f_2(m)$, hence the \defined{antimonotone} in the name.

We denote the resulting graph as $G_{ls2} \defeq (V_{lin}, E_{lin} \cup \set{(p_n, p_{f_2(n)}) \st n \in \setgeq{N}{2}})$, \cref{fig_ls2} shows an excerpt. For any $n$, we say it belongs to \defined{generation} $\floor{\log_2(n)}$. We say $p_{2^{t + 1} - 1}$ is the \defined{vertebra} of generation $t$, and $\bigcup_{k \leq t} \set{p_{2^{k + 1} - 1}}$ is the \defined{spine} of generation $t$.

\begin{figure}
  \centering
  \includegraphics{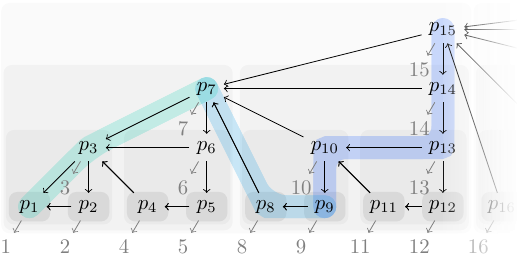}
  \caption{The simple antimonotone binary linking scheme, highlighting $\certificatepool(9)$, which consists of the paths \textcolor{cPoolPathText1}{from the next vertebra to $p_9$}, \textcolor{cPoolPathText2}{from $p_9$ to the previous vertebra}, and \textcolor{cPoolPathText3}{from the previous vertebra to $p_1$}.}
  \label{fig_ls2}
\end{figure}

Let $n$ a number of generation $t$. Then, the union of the shortest paths from the vertebra of $t$ to $p_n$, from $p_n$ to the vertebra of $t - 1$, and from that vertebra to $p_1$ (the latter two paths together form the shortest path from $p_n$ to $p_1$) is a certificate pool for $n$~\cite{buldas1998time}. We proceed with a proof sketch for the size of the corresponding positional certificates.

Observe the recursive structure of $G_{ls2}$: the graph of the first $t + 1$ generations consist of two copies of the graph of the first $t$ generations --- with the outgoing edges of the spine of the copy being replaced with edges to the (original) vertebra of generation $t$ --- and a new vertebra $p_{2^{(t+1)+1} - 1}$. This recursive structure enables convenient inductive reasoning based on the generation of a vertex.

Further observe that every vertex is either a vertebra vertex for some generation $t$, or a (transitive) copy of a vertebra vertex for some generation $t$. In both cases, we say that the vertex is of \defined{order} $t$.

For $n$ of generation $t$, the shortest path from the vertebra of $t$ to the vertebra of $t - 1$ via $p_n$ is of maximal size if $p_n$ is of order $0$. The path consists of vertices of decrementing order from the vertebra of $t$ to $p_n$, followed by vertices of incrementing order up to the vertebra of $t - 1$. Consequently, all order $0$ vertices of the same generation yield certficate pools of the same size: $2t - 1$. A proper proof would perform induction on $t$: removing the vertebrae of $t$ and $t - 1$ from the path yields a path isomorphic to that of $n - 2^{t} - 1$, which is of generation $t - 1$.

The out-neighborhood of this path has at least the same size, because every $p_n$ has an edge to $n$. More tricky are the edges to other $p_m$; we need to count the number of such edges that lead outside the path. The edges corresponding to $f_2$ never do so, this would immediately contradict antimonotonicity. It remains to consider the edges of the form $(p_n, p_{n - 1})$.

For vertices of generation at most $2$, all these edges are part of the path. For vertices of any other generation $t + 1$, there are $t - 2$ such edges leading outside the path, as can be seen inductively: if the successor of vertex $p_{2^{t} - 3}$ on the path is not its numeric predecessor, the number of predecessor edges outside the path increases by one compared to the path for the previous generation. If the successor of vertex $p_{2^{t} - 3}$ on the path \textit{is} its numeric predecessor, then $p_{f_2(2^{t} - 3)}$ is part of the path but its predecessor is not, again contributing exactly one additional edge to a path of the previous generation (see \cref{fig_ls2neighborhood}).

\begin{figure*}
  \centering
  \includegraphics{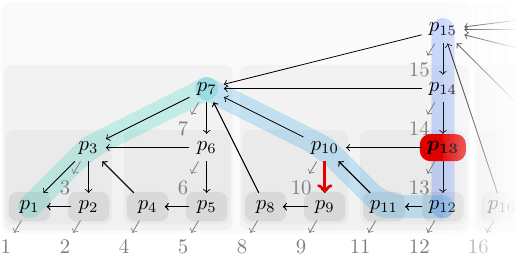}\includegraphics{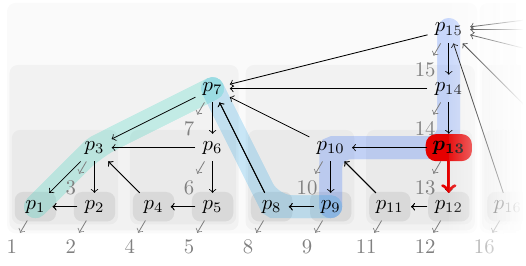}
  \caption{Visualizing the inductive step in counting the number of predecessor edges outside the certificate pool. Whether $n$ of generation $t$ is greater than or equal to $2^{t} + 2^{t - 1} - 1$ (i.e., the numeric predecessor of $p_{2^{t} - 3}$ is part of the certificate pool) or not, a single edge more leads outside the certificate pool than for a number of generation $t - 1$. The graphic shows the concrete case of $t \defeq 3$ with $n \defeq 12$ on the left and $n \defeq 9$ right. The crucial vertex is \textcolor{myRed}{$p_{13}$}.}
  \label{fig_ls2neighborhood}
\end{figure*}

For the full positional certificate of $n$ of generation $t$, it remains to add the size of the out-neighborhood of the shortest path from the vertebra of $t-1$ to $p_1$. This path has $t$ vertices, each contributing two vertices to the out-neighborhood ($p_{m - 1}$ and $m$), except for $p_1$, which has only a single out-neighbor. Adding everything up (and accounting for the double-counting of $p_{t - 1}$) yields the positional certificate size of $(5 \cdot \floor{\log_2(n)} - 3) \cdot k$, where $k$ is the size of an individual hash.

The shortest path from $p_n$ to $p_1$ is of logarithmic size and it can serve as a digest pool (in fact, antimonotonicity implies that this shortest path is a digest pool for \textit{every} antimonotone scheme). The observation that the graph has no edge which jumps over any vertebra can be used to further shrink the digest pool to the shortest path from $p_n$ to the vertebra of the previous generation.

The \defined{optimal antimonotone binary linking scheme}~\cite{buldas1998new} extends the linear scheme with a slightly different, antimonotone function $f_3(n)$, to obtain the graph $G_{ls3} \defeq (V_{lin}, E_{lin} \cup \set{(p_n, p_{(n)}) \st n \in \setgeq{N}{2}})$ (\cref{fig_ls3}):

\begin{align*}
  f_3(n) &\defeq \begin{cases}
    n - (3^{k - 1} + 1) &\mbox{if } n = \frac{3^k - 1}{2}, k \in \N \\
    n - (\frac{3^{h(n)} - 1}{2} + 1) & \mbox{otherwise}\\
  \end{cases}\\
  h(n) &\defeq \begin{cases}
    k &\mbox{if } n = \frac{3^k - 1}{2}, k \in \N \\
    h(n - \frac{3^{k - 1} - 1}{2}) & \mbox{if } \frac{3^{k - 1} - 1}{2} < n < \frac{3^k - 1}{2}, k \in \N
  \end{cases}\\
\end{align*}

\begin{figure*}
  \centering
  \includegraphics[width=18cm]{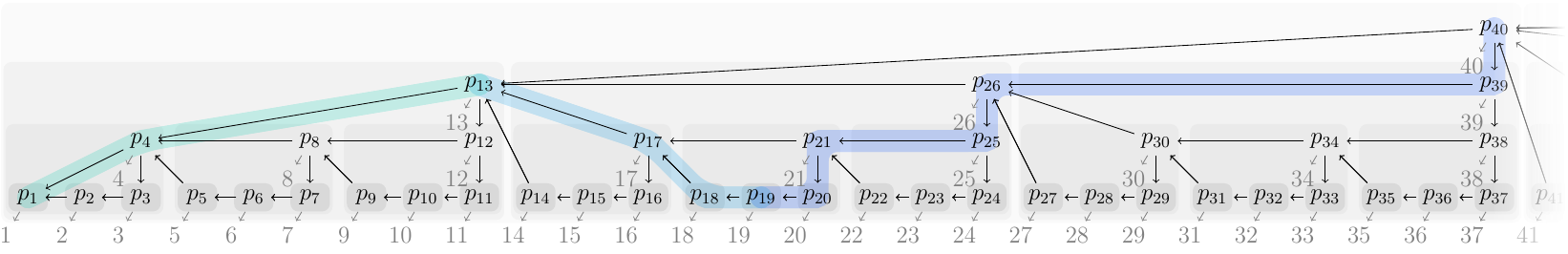}
  \caption{The optimal antimonotone binary linking scheme, highlighting $\certificatepool(19)$, which consists of the paths \textcolor{cPoolPathText1}{from the next vertebra to $p_19$}, \textcolor{cPoolPathText2}{from $p_19$ to the previous vertebra}, and \textcolor{cPoolPathText3}{from the previous vertebra to $p_1$}.}
  \label{fig_ls3}
\end{figure*}

Certificate pools (and their analysis) take the same shape as those of the simple antimonotone scheme, but with different generations and vertebra: the generation of $n$ in the optimal scheme is $\floor{\log_3(2n)}$, the vertebra of generation $t$ is $\frac{3^{t + 1} - 1}{2}$.

For $n$ of generation $t$, inductive arguments analogous to to those of the simple antimonotone scheme yield a maximal size of $3t + 1$ for the path from the vertebra of $t$ to the vertebra of $t - 1$ via $p_n$, and $2t - 3$ additional vertices for the out-neighborhood. The shortest path from the vertebra of $t-1$ to $p_1$ contributes another $2t - 1$ vertices to the positional certificate, yielding the total size of $(7 \cdot \floor{\log_3(2n)} - 4) \cdot k$ (again accounting for the double-counting of $p_{t - 1}$). This is more efficient than the simple antimonotone scheme; for all $n \geq 128$ we have $(7 \cdot \floor{\log_3(2n)} - 4) \cdot k < (5 \cdot \floor{\log_2(n)} - 3) \cdot k$.

The shortest path from $p_n$ to the vertebra of the previous generation can again serve as a digest pool. In the optimal antimonotone scheme, this path contains redundancies however. For example, consider $p_{24}$ in \cref{fig_ls3}: none of the labels of $p_{23}$, $p_{22}$, or $p_{17}$ are directly involved in the computation of the label of any greater vertex. For the digest pool of $n$, it suffices to take the $p_m$ with maximal $m$ of each order (with $m \leq n$ and for a maximal order of the generation of $n$).

We would like to point the interested reader to the elegant characterization of all antimonotone binary graphs~\cite{buldas1998new} that forms the basis of the optimality proof for the optimal antimonotone scheme amongst all antimonotone binary schemes in the setting of prefix authentication for strings of bounded length. All antimonotone binary graphs can be constructed from a graph product operation $\otimes$, starting from the trivial graph $G_1$. The antimotone product is an efficient way of thinking about the antimonotone schemes; the rather intimidating functions of natural numbers to describe the two schemes we presented turn into neat, immediately related one-liners: $G_{simple}^{i + 1} \defeq G_{simple}^i \otimes G_{simple}^i \otimes G_1$ and $G_{opt}^{i + 1} \defeq G_{opt}^i \otimes G_{opt}^i \otimes G_{opt}^i \otimes G_1$.

\subsection{Merkle Trees}

Whereas the schemes we considered so far are extensions of $G_{lin}$, the remaining schemes utilize Merkle trees. To unify their presentation, we first define the infinite Merkle tree $G_{tree}$ on which they build (\cref{fig_tree}) in isolation:

\begin{align*}
  V_{tree} &\defeq \bigl\{(n, k) \st \text{$n \in \N, k \in \Nz$ and $2^k \divides n$}\bigr\},\\
  E_{tree} &\defeq \Bigl\{\bigl((n_0, k + 1), (n_1, k)\bigr) \st \text{$n_0 = n_1$ or $n_0 = n_1 + 2^{k}$}\Bigr\}\\
  & \cup \Bigl\{\bigl((i, 0), i\bigr) \st i \in \N\Bigr\},\\
  G_{tree} &\defeq (V_{tree}, E_{tree}).
\end{align*}

\begin{figure}
  \centering
  \includegraphics{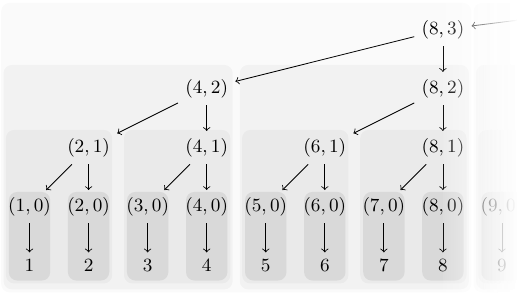}
  \caption{The start of the infinite Merkle tree $G_{tree}$.}
  \label{fig_tree}
\end{figure}

We can unify parts of the complexity analysis of the remaining schemes by analyzing $G_{tree}$. We first define the forest that corresponds to the first $n$ numbers: $\tpaston{G_{tree}}{n} \defeq G_{tree}[\setleq{V_{tree}}{n} \cup \set{(i, k) \st i \leq n}]$. $\tpaston{G_{tree}}{n}$ has at most $3n$ vertices for every $n$, but $\tpaston{G_{tree}}{n} - \tpaston{G_{tree}}{n - 1}$ can have up to $\ceil{\log_2(n)}$ vertices. Unlike the schemes we have seen so far, all schemes based on $G_{tree}$ thus require a non-constant amount of information for a single sequence item in the worst case.

We further define $\nextroot(n) \defeq (2^{\ceil{\log_2(n)}}, \ceil{\log_2(n)})$, the root of the smallest complete subtree to contain both $1$ and $n$, and $\nextpower(n) \defeq (2^{\ceil{\log_2(n)}}, 0)$.

The complexity analysis of several schemes depends on the number of roots in $\tpaston{G_{tree}}{n}$. Observe that the number of leaves of every tree in $\tpaston{G_{tree}}{n}$ is a power of two, and observe further that the trees in $\tpaston{G_{tree}}{n}$ all have different, strictly decreasing sizes. Every power of two less than $n$ occurs either once or not at all. In other words, the trees of $\tpaston{G_{tree}}{n}$ correspond directly to the binary representation of $n$.

\subsection{Threaded Authentication Trees}

We now turn to the first construction to use Merkle trees, the \defined{threaded authentication trees}~\cite{buldas2000optimally}. Threaded authentication trees start from $G_{tree}$ and then add edges from every $(n, 0)$ to the roots of the trees of $\tpaston{G_{tree}}{n}$, yielding a linking scheme with $\gcommit(n) \defeq (n, 0)$. The certificate pool of $n$ is the union of the shortest path from $\nextroot(n)$ to $(n, 0)$ and the shortest path from $\nextroot(n)$ to $(1, 0)$. See \cref{fig_tat} for an example.

\begin{figure}
  \centering
  \includegraphics{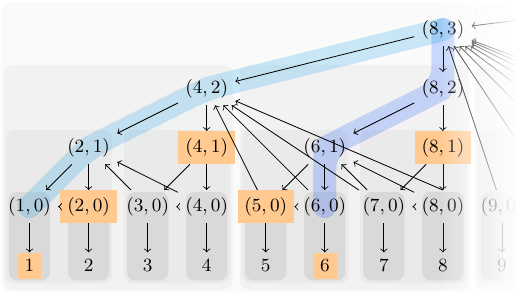}
  \caption{A threaded authentication tree, highlighting $\certificatepool(6)$, which consists of the paths \textcolor{cPoolPathText1}{from $\nextroot(6)$ to $(6, 0)$} and \textcolor{cPoolPathText2}{from $\nextroot(6)$ to $(1, 0)$}, and its \textcolor{cCertText}{out-neighborhood}.}
  \label{fig_tat}
\end{figure}

This definition of certificate pools yields positional certificates of size $2 \cdot \ceil{\log_2(n)} \cdot k$, where $k$ is the size of a single hash. This is almost twice as much as in the setting with rounds of known size, but it still outperforms the anti-monotone schemes and is optimal among all schemes we survey. But unlike the antimonotone schemes, the underlying graph is of super-linear size, $\tpaston{G_{tat}}{n}$ has $\complexity{n \log(n)}$ edges in the worst case:

Assume that $n = 2^k$. Remember that $(i, 0)$ has an outgoing edge for every $1$ in the binary representation of $i$. The numbers from $2^{k - 1}$ to $2^k$ (exclusively) are exactly the $\frac{n}{2}$ $k$-bit numbers. The total number of $1$ bits amongst them --- that is, the number of outgoing edges of the vertices $(\frac{n}{2}, 0), (\frac{n}{2} + 1, 0), \cdots, (n - 1, 0)$ --- is $\frac{n}{2} \cdot \frac{k}{2} \in \complexity{n \cdot k} = \complexity{n \log(n)}$.

As a consequence, checking a prefix certificate for two sequences of lengths $len_n < len_t$ can take time in $\complexity{\log(len_t) \cdot \log(\log(len_t))}$, despite the size of the certificate being in $\complexity{\log(len_t)}$. This is asymptotically slower than checking certificates for antimonotone schemes. The original presentation of threaded authentication trees focuses on certificate size only and does not mention this flaw.

As the outgoing edges of any vertex $(n, h)$ go to a root of a tree in $\tpaston{G_{tree}}{n}$, these roots give a digest pool for $n$. Since these roots correspond to binary digits, the size of the digest pool of $n$ is at most $\log_2(n)$.

\subsection{Hypercore}

Whereas threaded authentication trees turn $G_{tree}$ into a linking scheme, we now turn to approaches that turn $G_{tree}$ into more general TPASs. All these approaches share the insight that when $\tpaston{G_{tree}}{n}$ has a single root, that root can serve as a digest vertex for a sequence of $n$ items. But $G_{tree}$ has no possible digest vertices for sequences of any other length.

\defined{Hypercore}~\cite{ogden2017dat} takes a direct approach to augmenting $G_{tree}$ so that there is a digest vertex for every $n$. For every $n$ that is not a power of two, add a vertex $d_n$ with an outgoing edge to every root of $\tpaston{G_{tree}}{n}$ to obtain $G_{hyper}$ (\cref{fig_hypercore}).

\begin{figure}
  \centering
  \includegraphics{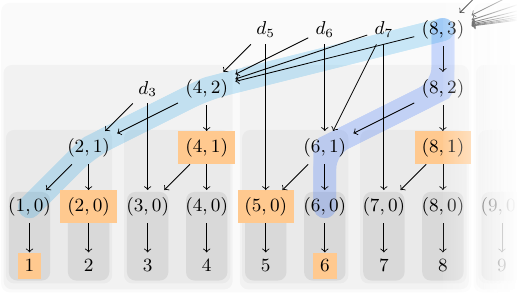}
  \caption{A hypercore, highlighting $\certificatepool(6)$, which consists of the paths \textcolor{cPoolPathText1}{from $\nextroot(6)$ to $(6, 0)$} and \textcolor{cPoolPathText2}{from $\nextroot(6)$ to $(1, 0)$}, and its \textcolor{cCertText}{out-neighborhood}.}
  \label{fig_hypercore}
\end{figure}

Then define $\gcommit(n)$ as $d_n$, or the root of $\tpaston{G_{tree}}{n}$ if $n$ is a power of two. $\dock(n)$ is then simply the set of roots of $\tpaston{G_{hyper}}{n}$, a set which \determines{} $\gcommit(n)$ by construction. Using these definitions of $\gcommit$ and $\dock$ leaves exactly one valid choice for defining $\gcertify(len_s, len_t)$: the unique family of paths that start in $\gcommit(len_t)$ and end ``just before'' $\dock(len_t)$.

The certificate pool of $n$ consists of the union of the path from $\nextroot(n)$ to $(n, 0)$ and the path from $\nextroot(n)$ to $(1, 0)$, just like with threaded authentication trees. Despite the different underlying graphs, the out-neighborhoods (and hence the sizes of positional certificates) are identical for threaded authentication trees and hypercores.

To see why the two paths yield valid certificate pools for hypercores, first observe that the out-neighborhood of the path from $\gcommit(n)$ to $(n, 0)$ is a subset of the out-neighborhood of the path from $\nextroot(n)$ to $(n, 0)$. 

Now consider any two numbers $len_s < len_t$. If $\nextroot(len_s) \neq \nextroot(len_t)$, then the path from $\nextroot(len_t)$ to $(1, 0)$ contains $\nextroot(len_s)$. And the out-neighborhood of the path from $\nextroot(len_s)$ to $(len_s, 0)$ contains $\dock(len_s)$, so together they contain $\gcertify(len_s, len_t)$.

In the other case of $\nextroot(len_s) = \nextroot(len_t)$, both $len_s$ and $len_t$ lie in the complete subtree that contains the leaves $2^{\ceil{\log_2(len_s)} - 1} + 1$ to $2^{\ceil{\log_2(len_s)}}$. This tree is isomorphic to the tree with the first $2^{\ceil{\log_2(len_s)} - 1}$ leaves, so $\certificatepool(len_s, len_t)$ is a valid certificate pool exactly if $\certificatepool(len_s - 2^{\ceil{\log_2(len_s)} - 1}, len_t - 2^{\ceil{\log_2(len_s)} - 1})$ is a valid certificate pool. Hence, validity follows by induction --- the base case is $len_s = 1$ and $len_t = 2$, for which the union of the certificate paths does contain $\gcertify(1, 2)$.

The roots of $\tpaston{G_{hyper}}{n}$ are a digest pool for $n$; The labels of the roots of $\tpaston{G_{hyper}}{n + 1}$ can be computed from the labels of the roots of $\tpaston{G_{hyper}}{n}$ together with sequence item $n + 1$.

Like threaded authentication trees, hypercores are of super-linear size: $\tpaston{G_{hyper}}{n}$ has $\complexity{n \cdot \log(n)}$ edges. But the number of edges within a prefix certificate is linear in the size of the certificate, so certificate validation lies in $\complexity{\log(n)}$.

Still, the super-linear size means that creating a sequence of length $n$ step-by-step takes $\complexity{n \log(n)}$ time, whereas the antimonotone schemes only need $\complexity{n}$ time. Furthermore, hypercores have identifiers of up to logarithmic size, unlike the constant-sized identifiers of antimonotone schemes or threaded authentication trees.

\subsection{Transparency Logs}

The \defined{transparency log} construction~\cite{laurie2014certificate} creates digest vertices for every $n$ by iteratively adding a parent vertex to the roots of the two smallest trees in $\tpaston{G_{tree}}{n}$ until there is a single root (\cref{fig_ct}). This root then serves as the digest of $n$.

\begin{figure}
  \centering
  \includegraphics{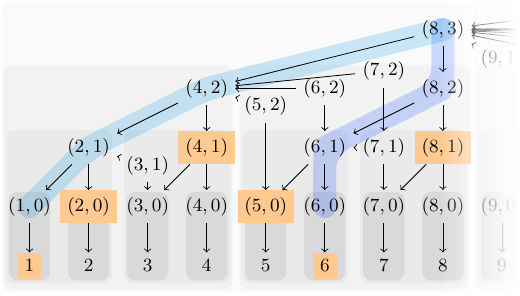}
  \caption{A certificate transparency log, highlighting $\certificatepool(6)$, which consists of the paths \textcolor{cPoolPathText1}{from $\nextroot(6)$ to $(6, 0)$} and \textcolor{cPoolPathText2}{from $\nextroot(6)$ to $(1, 0)$}, and its \textcolor{cCertText}{out-neighborhood}.}
  \label{fig_ct}
\end{figure}

Notice that contracting the newly created vertices for the same $n$ yields exactly the underlying graph of the hypercore scheme. The certificate transparency scheme is ultimately a deterministic (but essentially arbitrary) way of subdividing the digest vertices of hypercore until every vertex has at most two outgoing edges. Consequently, $\dock(n)$ is identical to that of hypercore, and $\gcertify(len_s, len_t)$ is again the unique family of paths that start in $\gcommit(len_t)$ and end ``just before'' $\dock(len_t)$.

Unlike hypercore, vertices in the certificate transparency scheme have a constant-bounded out-degree (two). The price is a super-linear number of vertices and the hashing of twice as many bits during certificate validation as hypercore. Otherwise, the same complexity analyses apply to both schemes.

In the context of certificate transparency, we want to point out that identifier sizes can matter outside timestamping. After submitting an entry to a certificate transparency log, the log operator replies\footnote{Setting aside CT-specific optimization details such as deliberate merge delays.} with the new \textit{signed tree head} (\textit{digest} in our terminology) and an \textit{inclusion proof} (an \textit{identifier} in our terminology). Using a scheme with constant-sized identifiers would be an asymptotic improvement over the certificate transparency scheme for this operation.

Five years prior to the publication of the certificate transparency scheme, Crosby and Wallach~\cite{crosby2009efficient} presented a highly similar scheme: whereas certificate transparency creates the smallest possible binary trees that contain all trees of $\tpaston{G_{tree}}{n}$, Crosby and Wallach's scheme also ensures that the vertices corresponding to sequence items all have the same height in the tree.

This makes the resulting graphs slightly larger supergraphs of the certificate transparency graphs, without any actual advantages. In hindsight, we can hence recommend to disregard Crosby and Wallach's scheme, while still appreciating that theirs is the first published non-linking scheme to solve prefix authentication.

\subsection{Summary}

Prior presentation of all the schemes we surveyed has focused exclusively on prefix or positional certificate sizes. From this limited perspective, threaded authentication trees, hypercores and certificate transparency logs are equivalent, and antimonotone schemes are inferior.

The picture drastically changes when taking into account the other complexity metrics of prefix authentication schemes. All three schemes with minimal certificates have a super-linear total size. Threaded authentication trees have smaller identifiers than hypercore and transparency logs, but suffer from super-logarithmic verification times. The scheme of Blibech and Gabillon~\cite{blibech2006new} does achieve the same certificate sizes with a linearly-sized underlying graph, but it does not generalize to the setting of unbounded sequence lengths.

\Cref{fig_summary} summarizes the results of our complexity analyses of the presented schemes.

\begin{table*}
  \centering
  \begin{tabular}{| c || c c c c |} 
   \hline
    & Linear & Full & Skiplist & Simple Antimonotone \\
   \hline\hline
   Positional Certificate & $n \cdot k$ & $n \cdot k$ & $\complexity{\log(n)^2} \cdot k$ & $(5 \cdot \floor{\log_2(n)} - 3) \cdot k$ \\ 
   Certificate Validation & $\complexity{certsize}$ & $\complexity{certsize^2}$ & $\complexity{certsize}$ & $\complexity{certsize}$ \\
   Edges Amortized & $\complexity{n}$ & $\complexity{n^2}$ & $\complexity{n}$ & $\complexity{n}$ \\
   Edges Worst Case & $\complexity{1}$ & $\complexity{n}$ & $\complexity{\log(n)}$ & $\complexity{1}$ \\
   Vertices Amortized & $\complexity{n}$ & $\complexity{n}$ & $\complexity{n}$ & $\complexity{n}$ \\
   Vertices Worst Case & $\complexity{1}$ & $\complexity{1}$ & $\complexity{1}$ & $\complexity{1}$ \\
   Identifier Amortized & $\complexity{1}$ & $\complexity{1}$ & $\complexity{1}$ & $\complexity{1}$ \\
   Identifier Worst Case & $\complexity{1}$ & $\complexity{1}$ & $\complexity{1}$ & $\complexity{1}$ \\
   Digest Pool & $1$ & $n$ & $\floor{\log_2(n)}$ & $\floor{\log_2(n)}$ \\ 
   \hline
    \hline
     & Optimal Antimonotone & Threaded Authentication & Hypercore & Transparency Log \\
    \hline\hline
    Position Certificate & $(7 \cdot \floor{\log_3(2n)} - 4) \cdot k$ & $2 \cdot \ceil{\log_2(n)} \cdot k$ & $2 \cdot \ceil{\log_2(n)} \cdot k$ & $2 \cdot \ceil{\log_2(n)} \cdot k$ \\ 
    Certificate Validation & $\complexity{certsize}$ & $\complexity{certsize \cdot \log(certsize)}$ & $\complexity{certsize}$ & $\complexity{certsize}$ \\
    Edges Amortized & $\complexity{n}$ & $\complexity{n \cdot \log(n)}$ & $\complexity{n \cdot \log(n)}$ & $\complexity{n \cdot \log(n)}$ \\
    Edges Worst Case & $\complexity{1}$ & $\complexity{\log(n)}$ & $\complexity{\log(n)}$ & $\complexity{\log(n)}$ \\
    Vertices Amortized & $\complexity{n}$ & $\complexity{n}$ & $\complexity{n}$ & $\complexity{n \cdot \log(n)}$ \\
    Vertices Worst Case & $\complexity{1}$ & $\complexity{\log(n)}$ & $\complexity{\log(n)}$ & $\complexity{\log(n)}$ \\
    Identifier Amortized & $\complexity{1}$ & $\complexity{1}$ & $\complexity{1}$ & $\complexity{1}$ \\
    Identifier Worst Case & $\complexity{1}$ & $\complexity{1}$ & $\complexity{\log(n)}$ & $\complexity{\log(n)}$ \\
    Digest Pool & $\floor{\log_3(2n)}$ & $\floor{\log_2(n)}$ & $\floor{\log_2(n)}$ & $\floor{\log_2(n)}$ \\ 
    \hline
   \end{tabular}
  \caption{Summary of the complexity analyses of all presented schemes. The number of edges and vertices determines how long it takes to create a sequence of $n$ items. Amortized complexities report the metric for $\tpaston{G}{n}$, worst-case complexities report the metric for $\tpaston{G}{n} \setminus \tpaston{G}{n - 1}.$}
  \label{fig_summary}
\end{table*}

\section{Conclusion}\label{conclusion}

Generalizing from secure time stamping and logging to prefix authentication has allowed us to transfer knowledge between results that have not been connected so far. The class of \textit{transitive prefix authentication schemes} serves as a tool to compactly and efficiently present and analyze prior results. The nuanced analysis shows that no existing approach is strictly superior to any other. We hope that future system designs will take into account all complexity criteria of prefix authentication schemes rather than latching on the first scheme with sub-linear prefix certificates that they lay eyes upon.

The main questions we leave open are questions of optimality. While threaded authentication trees have provably minimal positional certificates amongst linking schemes in a timestamping setting with rounds of known length, we did not transfer the optimality result to the setting of prefix authentication \textit{without} rounds. Similarly, we do not know whether the optimal antimonotone scheme for rounds of bounded length remains optimal amongst antimonotone schemes for prefix authentication without rounds.

Another open question is whether the optimal positional certificate sizes amongst \textit{linking} schemes are optimal amongst all \textit{transitive prefix authentication} schemes.

Our complexity analyses further surface a natural design challenge: that of finding a transitive prefix authentication scheme that achieves positional certificates of size $2 \cdot \ceil{\log_2(n)} \cdot k$ while having an underlying graph of linear size. Such a scheme would strictly outperform threaded authentication trees, hypercore, and certificate transparency logs.

Finally, we would like to emphasize again some of the limitations of prefix authentication schemes. First, the byzanthine fault-tolerant distributed ``append-\textit{only} log'' does not exist, as reacting to forks adds more expressivity to the abstract data type than just an append-operation. And second, ``proactive'' fork detection by placing logs onto a blockchain only shifts the problem to that of detecting forks in the blockchain, which might remain undetected for just as long as forks in a log when dealing with an equally powerful adversary (as it is literally the exact same problem).

The universe is cold and dark when it comes to distributed systems, but pretending otherwise always does more harm than good.

\bibliographystyle{plain}
\bibliography{\jobname}

\begin{thebibliography}{10}

\bibitem{al2018contour}
Mustafa Al-Bassam and Sarah Meiklejohn.
\newblock Contour: A practical system for binary transparency.
\newblock In {\em Data Privacy Management, Cryptocurrencies and Blockchain
  Technology}, pages 94--110. Springer, 2018.
\newblock \url{https://arxiv.org/pdf/1712.08427.pdf}.

\bibitem{basin2014arpki}
David Basin, Cas Cremers, Tiffany Hyun-Jin Kim, Adrian Perrig, Ralf Sasse, and
  Pawel Szalachowski.
\newblock Arpki: Attack resilient public-key infrastructure.
\newblock In {\em Proceedings of the 2014 ACM SIGSAC Conference on Computer and
  Communications Security}, pages 382--393, 2014.
\newblock
  \url{https://cispa.saarland/group/cremers/downloads/papers/ccsfp200s-cremersA.pdf}.

\bibitem{bayer1993improving}
Dave Bayer, Stuart Haber, and W~Scott Stornetta.
\newblock Improving the efficiency and reliability of digital time-stamping.
\newblock In {\em Sequences Ii}, pages 329--334. Springer, 1993.
\newblock
  \url{https://citeseerx.ist.psu.edu/document?repid=rep1&type=pdf&doi=fcc58b43fe133e98025bc616fbddd96393ae48c7}.

\bibitem{blibech2006new}
Kaouthar Blibech and Alban Gabillon.
\newblock A new timestamping scheme based on skip lists.
\newblock In {\em Computational Science and Its Applications-ICCSA 2006:
  International Conference, Glasgow, UK, May 8-11, 2006, Proceedings, Part III
  6}, pages 395--405. Springer, 2006.
\newblock
  \url{https://www.researchgate.net/profile/Alban-Gabillon/publication/221432970_A_New_Timestamping_Scheme_Based_on_Skip_Lists/links/0deec52aff664d7605000000/A-New-Timestamping-Scheme-Based-on-Skip-Lists.pdf}.

\bibitem{buldas1998new}
Ahto Buldas and Peeter Laud.
\newblock New linking schemes for digital time-stamping.
\newblock In {\em ICISC}, volume~98, pages 3--14. Citeseer, 1998.
\newblock
  \url{https://citeseerx.ist.psu.edu/document?repid=rep1&type=pdf&doi=d3a005fb546ff78abc6ee453af4ee91aa6267c50}.

\bibitem{buldas2002eliminating}
Ahto Buldas, Peeter Laud, and Helger Lipmaa.
\newblock Eliminating counterevidence with applications to accountable
  certificate management.
\newblock {\em Journal of Computer Security}, 10(3):273--296, 2002.
\newblock \url{https://research.cyber.ee/~peeter/research/JCS161.pdf}.

\bibitem{buldas1998time}
Ahto Buldas, Peeter Laud, Helger Lipmaa, and Jan Villemson.
\newblock Time-stamping with binary linking schemes.
\newblock In {\em Annual International Cryptology Conference}, pages 486--501.
  Springer, 1998.
\newblock
  \url{https://citeseerx.ist.psu.edu/document?repid=rep1&type=pdf&doi=f2390ec334bf99cb3d532bc16e05b6b201ad7115}.

\bibitem{buldas2000optimally}
Ahto Buldas, Helger Lipmaa, and Berry Schoenmakers.
\newblock Optimally efficient accountable time-stamping.
\newblock In {\em International workshop on public key cryptography}, pages
  293--305. Springer, 2000.
\newblock
  \url{https://citeseerx.ist.psu.edu/document?repid=rep1&type=pdf&doi=aa47957c6bc0e2c19b39aa644cb6b2e0c1defc83}.

\bibitem{chase2016transparency}
Melissa Chase and Sarah Meiklejohn.
\newblock Transparency overlays and applications.
\newblock In {\em Proceedings of the 2016 ACM SIGSAC Conference on Computer and
  Communications Security}, pages 168--179, 2016.
\newblock
  \url{https://discovery.ucl.ac.uk/id/eprint/10055892/1/transparency.pdf}.

\bibitem{chuat2015efficient}
Laurent Chuat, Pawel Szalachowski, Adrian Perrig, Ben Laurie, and Eran Messeri.
\newblock Efficient gossip protocols for verifying the consistency of
  certificate logs.
\newblock In {\em 2015 IEEE Conference on Communications and Network Security
  (CNS)}, pages 415--423. IEEE, 2015.
\newblock \url{https://arxiv.org/pdf/1511.01514.pdf}.

\bibitem{crosby2009efficient}
Scott~A Crosby and Dan~S Wallach.
\newblock Efficient data structures for tamper-evident logging.
\newblock In {\em USENIX security symposium}, pages 317--334, 2009.
\newblock
  \url{https://www.usenix.org/legacy/event/sec09/tech/full_papers/crosby.pdf}.

\bibitem{demers1987epidemic}
Alan Demers, Dan Greene, Carl Hauser, Wes Irish, John Larson, Scott Shenker,
  Howard Sturgis, Dan Swinehart, and Doug Terry.
\newblock Epidemic algorithms for replicated database maintenance.
\newblock In {\em Proceedings of the sixth annual ACM Symposium on Principles
  of distributed computing}, pages 1--12, 1987.
\newblock \url{https://dl.acm.org/doi/pdf/10.1145/41840.41841}.

\bibitem{dowling2016secure}
Benjamin Dowling, Felix G{\"u}nther, Udyani Herath, and Douglas Stebila.
\newblock Secure logging schemes and certificate transparency.
\newblock In {\em European Symposium on Research in Computer Security}, pages
  140--158. Springer, 2016.
\newblock \url{https://eprint.iacr.org/2016/452.pdf}.

\bibitem{eijdenberg2015verifiable}
Adam Eijdenberg, Ben Laurie, and Al~Cutter.
\newblock Verifiable data structures.
\newblock {\em Google Research, Tech. Rep}, 2015.
\newblock \url{https://continusec.com/static/VerifiableDataStructures.pdf}.

\bibitem{fahl2014hey}
Sascha Fahl, Sergej Dechand, Henning Perl, Felix Fischer, Jaromir Smrcek, and
  Matthew Smith.
\newblock Hey, nsa: Stay away from my market! future proofing app markets
  against powerful attackers.
\newblock In {\em proceedings of the 2014 ACM SIGSAC conference on computer and
  communications security}, pages 1143--1155, 2014.
\newblock \url{https://teamusec.de/pdf/conf-ccs-FahlDPFSS14.pdf}.

\bibitem{goodrich2000efficient}
Michael~T Goodrich and Roberto Tamassia.
\newblock Efficient authenticated dictionaries with skip lists and commutative
  hashing.
\newblock Technical report, Technical Report, Johns Hopkins Information
  Security Institute, 2000.
\newblock \url{https://cs.brown.edu/cgc/stms/papers/hashskip.pdf}.

\bibitem{haber1990time}
Stuart Haber and W~Scott Stornetta.
\newblock How to time-stamp a digital document.
\newblock In {\em Conference on the Theory and Application of Cryptography},
  pages 437--455. Springer, 1990.
\newblock \url{https://link.springer.com/content/pdf/10.1007/BF00196791.pdf}.

\bibitem{kim2013accountable}
Tiffany Hyun-Jin Kim, Lin-Shung Huang, Adrian Perrig, Collin Jackson, and
  Virgil Gligor.
\newblock Accountable key infrastructure (aki) a proposal for a public-key
  validation infrastructure.
\newblock In {\em Proceedings of the 22nd international conference on World
  Wide Web}, pages 679--690, 2013.
\newblock \url{https://dl.acm.org/doi/abs/10.1145/2488388.2488448}.

\bibitem{kubilay2019certledger}
Murat~Yasin Kubilay, Mehmet~Sabir Kiraz, and Hac{\i}~Ali Mantar.
\newblock Certledger: A new pki model with certificate transparency based on
  blockchain.
\newblock {\em Computers \& Security}, 85:333--352, 2019.
\newblock \url{https://arxiv.org/pdf/1806.03914.pdf}.

\bibitem{laurie2014certificate}
Ben Laurie.
\newblock Certificate transparency.
\newblock {\em Communications of the ACM}, 57(10):40--46, 2014.
\newblock \url{https://dl.acm.org/doi/fullHtml/10.1145/2659897}.

\bibitem{laurie2014certificateb}
Ben Laurie.
\newblock Certificate transparency: Public, verifiable, append-only logs.
\newblock {\em Queue}, 12(8):10--19, 2014.
\newblock \url{https://dl.acm.org/doi/pdf/10.1145/2668152.2668154}.

\bibitem{laurie2012revocation}
Ben Laurie and Emilia Kasper.
\newblock Revocation transparency.
\newblock {\em Google Research, September}, 33, 2012.
\newblock \url{http://www.links.org/files/RevocationTransparency.pdf}.

\bibitem{rfc9162}
Ben Laurie, Adam Langley, Emilia Kasper, Eran Messeri, and Rob Stradling.
\newblock {Certificate Transparency Version 2.0}.
\newblock RFC 9162, December 2021.
\newblock \url{https://www.rfc-editor.org/info/rfc9162}.

\bibitem{leibowitz2021ctng}
Hemi Leibowitz, Haitham Ghalwash, Ewa Syta, and Amir Herzberg.
\newblock Ctng: Secure certificate and revocation transparency.
\newblock {\em Cryptology ePrint Archive}, 2021.
\newblock \url{https://eprint.iacr.org/2021/818.pdf}.

\bibitem{li2004secure}
Jinyuan Li, Maxwell~N Krohn, David Mazieres, and Dennis~E Shasha.
\newblock Secure untrusted data repository (sundr).
\newblock In {\em Osdi}, volume~4, pages 9--9, 2004.
\newblock
  \url{https://www.usenix.org/legacy/event/osdi04/tech/full_papers/li_j/li_j.pdf}.

\bibitem{madala2018certificate}
DSV Madala, Mahabir~Prasad Jhanwar, and Anupam Chattopadhyay.
\newblock Certificate transparency using blockchain.
\newblock In {\em 2018 IEEE International Conference on Data Mining Workshops
  (ICDMW)}, pages 71--80. IEEE, 2018.
\newblock \url{https://eprint.iacr.org/2018/1232.pdf}.

\bibitem{melara2015coniks}
Marcela~S Melara, Aaron Blankstein, Joseph Bonneau, Edward~W Felten, and
  Michael~J Freedman.
\newblock $\{$CONIKS$\}$: Bringing key transparency to end users.
\newblock In {\em 24th USENIX Security Symposium (USENIX Security 15)}, pages
  383--398, 2015.
\newblock
  \url{https://www.usenix.org/system/files/conference/usenixsecurity15/sec15-paper-melara.pdf}.

\bibitem{menezes2018handbook}
Alfred~J Menezes, Paul~C Van~Oorschot, and Scott~A Vanstone.
\newblock {\em Handbook of applied cryptography}.
\newblock CRC press, 2018.
\newblock
  \url{http://labit501.upct.es/~fburrull/docencia/SeguridadEnRedes/old/teoria/bibliography/HandbookOfAppliedCryptography_AMenezes.pdf}.

\bibitem{merkle1989certified}
Ralph~C Merkle.
\newblock A certified digital signature.
\newblock In {\em Conference on the Theory and Application of Cryptology},
  pages 218--238. Springer, 1989.
\newblock
  \url{https://link.springer.com/content/pdf/10.1007/0-387-34805-0_21.pdf}.

\bibitem{naor2000certificate}
Moni Naor and Kobbi Nissim.
\newblock Certificate revocation and certificate update.
\newblock {\em IEEE Journal on selected areas in communications},
  18(4):561--570, 2000.
\newblock
  \url{https://www.usenix.org/legacy/publications/library/proceedings/sec98/full_papers/nissim/nissim.pdf}.

\bibitem{nikitin2017chainiac}
Kirill Nikitin, Eleftherios Kokoris-Kogias, Philipp Jovanovic, Nicolas Gailly,
  Linus Gasser, Ismail Khoffi, Justin Cappos, and Bryan Ford.
\newblock $\{$CHAINIAC$\}$: Proactive $\{$Software-Update$\}$ transparency via
  collectively signed skipchains and verified builds.
\newblock In {\em 26th USENIX Security Symposium (USENIX Security 17)}, pages
  1271--1287, 2017.
\newblock
  \url{https://www.usenix.org/system/files/conference/usenixsecurity17/sec17-nikitin.pdf}.

\bibitem{ietf-trans-gossip-05}
Linus Nordberg, Daniel~Kahn Gillmor, and Tom Ritter.
\newblock {Gossiping in CT}.
\newblock Internet-Draft draft-ietf-trans-gossip-05, Internet Engineering Task
  Force, January 2018.
\newblock \url{https://datatracker.ietf.org/doc/draft-ietf-trans-gossip/05/}.

\bibitem{ogden2017dat}
Maxwell Ogden, Karissa McKelvey, Mathias~Buus Madsen, et~al.
\newblock Dat --- distributed dataset synchronization and versioning.
\newblock {\em Open Science Framework}, 10, 2017.
\newblock
  \url{https://terrymarine.com/wp-content/uploads/2017/07/724d7267d90052778b0530807512474b.pdf}.

\bibitem{papamanthou2008authenticated}
Charalampos Papamanthou, Roberto Tamassia, and Nikos Triandopoulos.
\newblock Authenticated hash tables.
\newblock In {\em Proceedings of the 15th ACM conference on Computer and
  communications security}, pages 437--448, 2008.
\newblock
  \url{https://citeseerx.ist.psu.edu/document?repid=rep1&type=pdf&doi=0556a7a0ceec16986a1817333de62149548b95ad}.

\bibitem{pugh1990skip}
William Pugh.
\newblock Skip lists: a probabilistic alternative to balanced trees.
\newblock {\em Communications of the ACM}, 33(6):668--676, 1990.
\newblock \url{https://dl.acm.org/doi/pdf/10.1145/78973.78977}.

\bibitem{pulls2015balloon}
Tobias Pulls and Roel Peeters.
\newblock Balloon: A forward-secure append-only persistent authenticated data
  structure.
\newblock In {\em European Symposium on Research in Computer Security}, pages
  622--641. Springer, 2015.
\newblock \url{https://eprint.iacr.org/2015/007.pdf}.

\bibitem{pulls2013distributed}
Tobias Pulls, Roel Peeters, and Karel Wouters.
\newblock Distributed privacy-preserving transparency logging.
\newblock In {\em Proceedings of the 12th ACM workshop on Workshop on privacy
  in the electronic society}, pages 83--94, 2013.
\newblock
  \url{https://www.esat.kuleuven.be/cosic/publications/article-2373.pdf}.

\bibitem{ryan2013enhanced}
Mark~D Ryan.
\newblock Enhanced certificate transparency and end-to-end encrypted mail.
\newblock {\em Cryptology ePrint Archive}, 2013.
\newblock \url{https://eprint.iacr.org/2013/595.pdf}.

\bibitem{schneier1999secure}
Bruce Schneier and John Kelsey.
\newblock Secure audit logs to support computer forensics.
\newblock {\em ACM Transactions on Information and System Security (TISSEC)},
  2(2):159--176, 1999.
\newblock \url{https://dl.acm.org/doi/pdf/10.1145/317087.317089}.

\bibitem{singh2017certificate}
Abhishek Singh, Binanda Sengupta, and Sushmita Ruj.
\newblock Certificate transparency with enhancements and short proofs.
\newblock In {\em Information Security and Privacy: 22nd Australasian
  Conference, ACISP 2017, Auckland, New Zealand, July 3--5, 2017, Proceedings,
  Part II 22}, pages 381--389. Springer, 2017.
\newblock \url{https://arxiv.org/pdf/1704.04937.pdf}.

\bibitem{syta2016keeping}
Ewa Syta, Iulia Tamas, Dylan Visher, David~Isaac Wolinsky, Philipp Jovanovic,
  Linus Gasser, Nicolas Gailly, Ismail Khoffi, and Bryan Ford.
\newblock Keeping authorities" honest or bust" with decentralized witness
  cosigning.
\newblock In {\em 2016 IEEE Symposium on Security and Privacy (SP)}, pages
  526--545. Ieee, 2016.
\newblock \url{https://arxiv.org/pdf/1503.08768.pdf}.

\bibitem{tamassia2003authenticated}
Roberto Tamassia.
\newblock Authenticated data structures.
\newblock In {\em European symposium on algorithms}, pages 2--5. Springer,
  2003.
\newblock
  \url{https://hashingit.com/elements/research-resources/2003-Tamassia-ADS.pdf}.

\bibitem{tarr2019secure}
Dominic Tarr, Erick Lavoie, Aljoscha Meyer, and Christian Tschudin.
\newblock Secure scuttlebutt: An identity-centric protocol for subjective and
  decentralized applications.
\newblock In {\em Proceedings of the 6th ACM conference on information-centric
  networking}, pages 1--11, 2019.
\newblock \url{https://dl.acm.org/doi/pdf/10.1145/3357150.3357396}.

\bibitem{tomescu2019transparency}
Alin Tomescu, Vivek Bhupatiraju, Dimitrios Papadopoulos, Charalampos
  Papamanthou, Nikos Triandopoulos, and Srinivas Devadas.
\newblock Transparency logs via append-only authenticated dictionaries.
\newblock In {\em Proceedings of the 2019 ACM SIGSAC Conference on Computer and
  Communications Security}, pages 1299--1316, 2019.
\newblock \url{https://dl.acm.org/doi/pdf/10.1145/3319535.3345652}.

\bibitem{tomescu2017catena}
Alin Tomescu and Srinivas Devadas.
\newblock Catena: Efficient non-equivocation via bitcoin.
\newblock In {\em 2017 IEEE Symposium on Security and Privacy (SP)}, pages
  393--409. IEEE, 2017.
\newblock
  \url{https://dspace.mit.edu/bitstream/handle/1721.1/137544/catena.pdf?sequence=2}.

\bibitem{wang2020blockchain}
Ze~Wang, Jingqiang Lin, Quanwei Cai, Qiongxiao Wang, Daren Zha, and Jiwu Jing.
\newblock Blockchain-based certificate transparency and revocation
  transparency.
\newblock {\em IEEE Transactions on Dependable and Secure Computing}, 2020.
\newblock \url{https://fc18.ifca.ai/bitcoin/papers/bitcoin18-final29.pdf}.

\bibitem{west2001introduction}
Douglas~Brent West et~al.
\newblock {\em Introduction to graph theory}, volume~2.
\newblock Prentice hall Upper Saddle River, 2001.
\newblock \url{https://faculty.math.illinois.edu/~west/igt/igtpref.ps}.

\bibitem{yu2016dtki}
Jiangshan Yu, Vincent Cheval, and Mark Ryan.
\newblock Dtki: A new formalized pki with verifiable trusted parties.
\newblock {\em The Computer Journal}, 59(11):1695--1713, 2016.
\newblock \url{https://arxiv.org/pdf/1408.1023.pdf}.

\bibitem{yumerefendi2007strong}
Aydan~R Yumerefendi and Jeffrey~S Chase.
\newblock Strong accountability for network storage.
\newblock {\em ACM Transactions on Storage (TOS)}, 3(3):11--es, 2007.
\newblock
  \url{https://www.usenix.org/legacy/event/fast07/tech/full_papers/yumerefendi/yumerefendi.pdf}.

\end{thebibliography}

\end{document}